\documentclass[twoside]{article}
\usepackage[accepted]{aistats2027}
\usepackage{amsmath,amssymb,amsthm,bm,graphicx}
\usepackage[round]{natbib}
\usepackage{booktabs}
\usepackage{tikz}
\usetikzlibrary{arrows.meta,calc,decorations.pathreplacing}
\usepackage[hidelinks]{hyperref}
\definecolor{iidcv}{HTML}{3A7CA5}
\definecolor{ours}{HTML}{FF7F0E}
\definecolor{solve}{HTML}{7F7F7F}

\newtheorem{theorem}{Theorem}
\newtheorem{proposition}{Proposition}
\newtheorem{lemma}[theorem]{Lemma}
\newtheorem{corollary}[theorem]{Corollary}
\theoremstyle{remark}
\newtheorem{remark}{Remark}
\theoremstyle{plain}
\newcommand{\MMD}{\mathrm{MMD}}
\newcommand{\E}{\mathbb{E}}
\newcommand{\ip}[2]{\langle #1, #2 \rangle}

\makeatletter
\renewcommand{\Notice@String}{Preprint.}
\makeatother

\begin{document}

\twocolumn[
\aistatstitle{Amortized quadrature for posterior expectations in inverse
problems}
\aistatsauthor{Ali Siahkoohi}
\aistatsaddress{Department of Computer Science, University of Central Florida}
]
\makeatletter
\gdef\@runningtitle{Amortized quadrature for posterior expectations in inverse problems}
\makeatother

\begin{abstract}
Uncertainty in the solution of an inverse problem and in the tasks performed
on it is quantified by posterior expectations, each an average of an
integrand over $M$ posterior samples. While designed quadratures improve on
the $\mathcal{O}(M^{-1/2})$ error of Monte-Carlo estimation, they solve an
optimization problem, often against the posterior density, for every new
observation, which can be computationally costly. To address this limitation, we introduce the \emph{quadrature
field}, a set-equivariant network that maps an observation and its $M$
posterior samples to an $M$-node signed-weight quadrature in one forward
pass. Trained once on a family of posteriors to minimize the worst-case
integration error over a class of functions, it serves any observation, any
$M$ and any integrand in that class with no further optimization. We show
that, with high probability and up to a computable slack, the resulting
quadrature is never worse than the Monte-Carlo estimate built from the
same samples. We validate the
quadrature field on closed-form and on learned posteriors, one constrained
by a partial differential equation, where it improves on the Monte-Carlo
estimate in median at every node count, often by orders of magnitude.
\end{abstract}
\section{Introduction}
\label{sec:intro}

Quantifying uncertainty in the solution of an inverse problem is of
fundamental importance in many areas of science and
engineering~\citep{tarantola2005inverse,stuart2010ip}. A Bayesian
formulation characterizes the solution by a posterior distribution, and a
practitioner computes expectations under it---e.g., a mean, a variance, or a
downstream quantity of interest---as weighted averages of an integrand over
$M$ posterior samples. Evaluating the integrand at a sample is often
expensive---a simulation of the downstream task, or a solve of the governing
partial differential equation (PDE)---so $M$
is the budget a practitioner can afford. Increasingly, the posterior is itself
an amortized model trained once for the whole family of
observations~\citep{papamakarios2016npe,greenberg2019snpe,cranmer2020sbi,BaldassariSiahkoohiGarnierEtAl_2023,SiahkoohiRizzutiOrozcoEtAl_2023,baptista2024monotone}.
Posterior samples are then cheap, and the integrand is the cost.

\begin{figure}[t]
\centering
\includegraphics[width=\columnwidth]{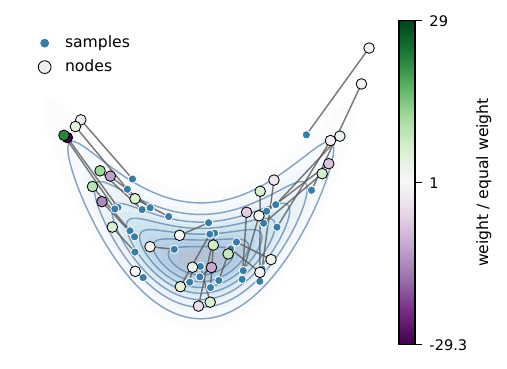}
\caption{A forward pass on a banana posterior: each sample joined to its
node, colored by its weight. Code:
\texttt{github.com/luqigroup/qfield}.}
\label{fig:teaser}
\end{figure}

Monte-Carlo estimation weights these samples equally, and its worst-case
error over the unit ball of a reproducing kernel Hilbert space---the maximum
mean discrepancy~\citep{gretton2012mmd}---decays as $M^{-1/2}$, the minimax
rate for rules computed from the samples
alone~\citep{tolstikhin2017minimax,dwivedi2024kernelthinning}. Designed
quadratures reach a lower error by reading the posterior beyond its samples:
its kernel
mean~\citep{ohagan1991bq,rasmussen2003bmc,chen2010herding,huszar2012sbq,briol2019pi},
its density~\citep{belhadji2026msip} or a large sample of
it~\citep{belhadji2026datadriven}, its score~\citep{riabiz2022steinthinning},
or a pool of samples far larger than
$M$~\citep{hayakawa2022recombination,dwivedi2024kernelthinning}. While these
constructions improve considerably on Monte-Carlo estimation, each runs once per
observation, a crucial limitation when an amortized posterior
serves a stream of observations.

To address this limitation, we introduce the \emph{quadrature field}, a
set-equivariant network~\citep{zaheer2017deepsets,lee2019settransformer}
trained once on a family of posteriors, which maps an observation and its
$M$ posterior samples to an $M$-node signed-weight quadrature in one forward
pass (Figure~\ref{fig:teaser}). Aside from the samples, the field reads the posterior only through its
kernel mean at the candidate nodes, available in closed form for Gaussian and
mixture posteriors and estimated from further samples otherwise, so neither a
density nor a score is required. The field displaces each sample to a node
and solves the weights at those nodes in closed form, differentiating through
the solve during training. Since the maximum mean discrepancy is the
worst-case error over the kernel's unit ball and the node count enters the
field as an input~\citep{perez2018film}, one field serves any observation,
any $M$, and any integrand in the ball with no optimization per observation.

In order to keep amortization from trading the guarantee away, we let the
kernel mean score every candidate: the field returns the better of its
output and the samples it received. We show that the returned quadrature is
never worse than the Monte-Carlo estimate, deterministically when the
kernel mean is exact and with high probability when it is estimated. This
guarantee extends the classical inequality for optimal weights at fixed
nodes~\citep[eq.~3]{belhadji2026datadriven} to moved nodes, any network,
and observations the field was not trained on.

We validate the quadrature field on closed-form posteriors and on learned
posteriors for tomography, groundwater flow, and a Lotka--Volterra
predator--prey family. The returned quadrature's median error falls below
the Monte-Carlo error at every node count. On the groundwater problem the
same expectation is reached with fewer forward solves at nearly every node
count and functional. On the Lotka--Volterra benchmark of kernel and Stein
thinning~\citep{riabiz2022steinthinning}, our method improves on kernel
thinning~\citep{dwivedi2024kernelthinning}, Stein thinning,
herding~\citep{chen2010herding}, and
recombination~\citep{hayakawa2022recombination}. It also reaches a
per-observation greedy search with Bayesian-quadrature
weights~\citep{huszar2012sbq} at a cost that does not grow with the pool
the search selects from. The amortization is thus not paid for in
accuracy.
\section{Related work}
\label{sec:related}

The quadrature field amortizes a construction built once per target.
Bayesian and kernel quadrature provide the object---a weighted point
set---and its closed-form weights at fixed
nodes~\citep{ohagan1991bq,rasmussen2003bmc,briol2019pi}, greedy selection
and kernel herding are its sequential
forms~\citep{chen2010herding,huszar2012sbq}, and control functionals and
Stein points reach the same end from the
score~\citep{oates2017controlfunctionals,chen2018steinpoints}.

The route we build on obtains nodes and weights together, as the stationary
configuration of an interacting particle system that descends the maximum
mean discrepancy---from the density in one paper and from a large sample of
the target in the other~\citep{belhadji2026msip,belhadji2026datadriven}. That
work states as background that the weights minimizing the discrepancy at
fixed nodes do no worse than equal weights
there~\citep[eq.~3]{belhadji2026datadriven}. We state in
Section~\ref{sec:guarantee} what this paper adds to it. Both run the descent once
per target.

While distribution compression shares our motivation, since each retained
point is one expensive downstream simulation~\citep{riabiz2022steinthinning},
it returns a different object. Kernel thinning and its near-linear form, with
coreset sizes within logarithmic factors of the lower
bound~\citep{phillips2020coresets,shetty2022compress,dwivedi2024kernelthinning},
recombination~\citep{hayakawa2022recombination}, leverage-score
sampling~\citep{chatalic2025leverage}, support points~\citep{mak2018support}
and the score-consuming Stein thinning~\citep{riabiz2022steinthinning} each
summarize a large sample by a small one. The quadrature field is conditional
on the observation, runs in one pass, reads $M$ samples and a kernel-mean
estimate, not a pool to thin, returns signed unit-sum weights, and comes with
a relative guarantee instead of a rate. We compare against kernel thinning,
Stein thinning and recombination, and against the greedy search and herding,
at matched node count and at matched information in
Section~\ref{sec:experiments}.

Amortization has reached integration by other routes, each amortizing an
estimator built for chosen integrands: learned proposals for a named
expectation~\citep{golinski2019amci}, transfer across related
integrals~\citep{xi2018related}, and control variates meta-learned across
tasks~\citep{sun2023meta} or trained conditionally on the
observation~\citep{SiahkoohiOh_2026}. We amortize the quadrature itself,
whose solved weights serve every integrand in the kernel's unit ball at once.
The posterior it integrates is what amortized variational inference and
neural posterior estimation already
produce~\citep{papamakarios2016npe,greenberg2019snpe,cranmer2020sbi,radev2022bayesflow,SiahkoohiRizzutiOrozcoEtAl_2023,baptista2024monotone},
and we fix in Section~\ref{sec:quadrature} what we assume of it.
\section{Kernel quadrature}
\label{sec:quadrature}

\subsection{Worst-case error and the kernel mean}
\label{sec:worstcase}

Let $\rho = \rho(\cdot \mid \bm{y})$ be the reference posterior on
$\mathbb{R}^d$ for the observation $\bm{y}$, and let
$Q = \sum_{j=1}^{M} w_j \delta_{\bm{z}_j}$ be a quadrature with nodes
$\bm{z}_j$ and weights $w_j$ that sum to one. Its error on an integrand $f$
is $|\sum_j w_j f(\bm{z}_j) - \int f \,\mathrm{d}\rho|$. Over the unit ball
of the reproducing kernel Hilbert space of a kernel $k$, the worst case of
this error is the maximum mean discrepancy~\citep{gretton2012mmd,briol2019pi},
\begin{equation}
\label{eq:mmd}
\mathrm{MMD}^2(Q, \rho)
= \sum_{i,j} w_i w_j k(\bm{z}_i, \bm{z}_j)
- 2 \sum_j w_j \mu_\rho(\bm{z}_j) + c_\rho ,
\end{equation}
with $\mu_\rho(\bm{x}) = \int k(\bm{x}, \bm{x}')\,\mathrm{d}\rho(\bm{x}')$
the kernel mean of $\rho$ and $c_\rho = \iint k \,\mathrm{d}\rho\,\mathrm{d}\rho$
its self-affinity.

The self-affinity is common to all quadratures of $\rho$ and cancels from
any comparison between two of them. Scoring a candidate quadrature therefore requires the kernel mean at its nodes and nothing else about $\rho$. The density
and the score never enter. The kernel mean is closed form for a Gaussian or a
Gaussian mixture under the squared-exponential kernel and a finite sum when
the reference is a dataset of samples. When the reference is a trained model, we
estimate the kernel mean from further samples of the model.

\begin{remark}[The reference is what is integrated]
\label{rem:reference}
Every guarantee in this paper is against $\rho$ as supplied. When $\rho$ is
a trained approximation of the posterior $\pi(\cdot \mid \bm{y})$, the
distance from $\rho$ to $\pi$ is a separate error, and the triangle
inequality bounds the error against $\pi$ by their sum.
\end{remark}

\paragraph{Kernel and coordinates.}
Throughout we use the squared-exponential kernel, so $k(\bm{x},\bm{x}) = 1$,
with the bandwidth set by the median heuristic on samples of the reference
(Appendix~\ref{app:setup}).
In high ambient dimension the off-diagonal entries of a sample's Gram matrix
concentrate at one value, and moving nodes buys nothing. On the tomography
problems, whose images have hundreds to thousands of pixels, the quadrature
is therefore trained and scored in a low-dimensional subspace, built from the
forward operator~\citep{cui2014lis,zahm2022certified} where the posterior is
exact and from samples of the reference where it is a trained flow, and the
nodes are lifted back for reconstructions. We state and test the guarantees
of Section~\ref{sec:guarantee} in the coordinates the quadrature is scored in.

\subsection{Monte Carlo estimation}
\label{sec:mc}

The Monte-Carlo estimate is the quadrature $Q_0$ whose nodes are $M$
independent samples of $\rho$ and whose weights are $1/M$. For a kernel with
unit diagonal, two lines of algebra give
\begin{equation}
\label{eq:floor}
\mathbb{E}\, \mathrm{MMD}^2(Q_0, \rho) = \frac{1 - c_\rho}{M} .
\end{equation}
The unit diagonal enters only through the number one: for any bounded kernel
the same two lines give $(\mathbb{E}\,k(X,X) - c_\rho)/M$
(Proposition~\ref{prop:floor-general}). We call \eqref{eq:floor} the
\emph{floor} and show the discrepancies in Section~\ref{sec:experiments} relative to it, so that below one is better. The identity is in the mean, and
the levels we report are medians.

The order $M^{-1}$ is a barrier for rules computed from the samples alone: no
estimator built from $M$ independent samples of the target, weighted or not,
improves on it uniformly over
targets~\citep{tolstikhin2017minimax,dwivedi2024kernelthinning}. Going below
this order takes information beyond the samples. The kernel mean is that
information, and we consume it in Section~\ref{sec:field}.
\section{The quadrature field}
\label{sec:field}

The quadrature field maps an observation and $M$ samples of the reference to
$M$ nodes, and a closed-form solve provides the weights there
(Figure~\ref{fig:architecture}). A comparison with the samples it received
then decides what is returned.

\begin{figure*}[t]
\centering
\begin{tikzpicture}[
  font=\footnotesize,
  >={Stealth[length=2.2mm]},
  box/.style={line width=1.2pt, rounded corners=4pt, align=center, inner sep=6pt},
  cond/.style={draw=black!45, fill=black!4, line width=0.8pt, rounded corners=3pt,
               align=center, inner sep=3.5pt},
  bigflow/.style={->, line width=1.6pt, draw=black!75},
  tap/.style={->, line width=0.9pt, draw=black!45, densely dashed},
]

\coordinate (seedc) at (1.85,0);
\draw[iidcv!55, line width=0.8pt] ($(seedc)+(-0.85,0)$) -- ($(seedc)+(0.85,0)$);
\foreach \x in {-0.68,-0.34,0,0.34,0.68}{
  \draw[iidcv, line width=1.2pt] ($(seedc)+(\x,0)$) -- ($(seedc)+(\x,0.38)$);
  \node[circle, draw=iidcv, fill=white, line width=0.7pt, minimum size=3.2pt,
        inner sep=0pt] at ($(seedc)+(\x,0.38)$){};
}
\node[anchor=south, iidcv, align=center] at ($(seedc)+(0,0.68)$)
  {$M$ samples of $\rho(\cdot\mid\bm{y})$};
\node[anchor=north, font=\scriptsize\itshape, text=iidcv!85, align=center]
  at ($(seedc)+(0,-0.14)$) {equal weights};

\node[box, draw=ours, fill=ours!8] (net) at (6.25,0.12)
  {\textbf{\textcolor{ours}{quadrature field $\Delta_\theta$}}\\[2pt]
   self-attention among samples\\[1pt]
   cross-attention to the observation};

\node[box, draw=solve, fill=solve!8] (slv) at (11.55,0.12)
  {\textbf{\textcolor{solve}{closed-form weights}}\\[2pt]
   unit-sum minimizer \eqref{eq:weights}};

\coordinate (outc) at (15.20,0);
\draw[ours!50, line width=0.8pt] ($(outc)+(-0.95,0)$) -- ($(outc)+(0.95,0)$);
\foreach \x/\h in {-0.86/0.52, -0.58/0.24, 0.10/-0.24, 0.26/0.44, 0.80/-0.13}{
  \draw[ours, line width=1.4pt] ($(outc)+(\x,0)$) -- ($(outc)+(\x,\h)$);
  \node[circle, draw=ours!85, fill=ours, minimum size=3.6pt, inner sep=0pt]
       at ($(outc)+(\x,\h)$){};
}
\node[anchor=south, ours, align=center] at ($(outc)+(0,0.72)$)
  {the moved quadrature $Q_{\mathrm{mv}}$};
\node[anchor=north, font=\scriptsize\itshape, text=ours!85, align=center]
  at ($(outc)+(0,-0.14)$) {signed, summing to one};

\draw[bigflow] ($(seedc)+(1.0,0.12)$) -- (net.west);
\draw[bigflow] (net.east) -- (slv.west);
\draw[bigflow] (slv.east) -- ($(outc)+(-1.1,0.12)$);
\node[cond, anchor=north] (yf) at (4.05,-1.15)
  {the observation};
\node[cond, anchor=north] (fl) at (7.55,-1.15)
  {the node count $M$};
\node[cond, anchor=north] (km) at (11.55,-1.15)
  {kernel mean of $\rho(\cdot\mid\bm{y})$};
\draw[tap] (yf.north) -- ($(net.south)+(-2.2,0)$);
\draw[tap] (fl.north) -- ($(net.south)+(1.3,0)$);
\draw[tap] (km.north) -- (slv.south);

\end{tikzpicture}
\caption{One forward pass of the quadrature field. The network displaces the
samples and a closed-form solve gives the weights. Stems are weights,
positive up and negative down. Dashed arrows are inputs.}
\label{fig:architecture}
\end{figure*}

\subsection{Architecture}
\label{sec:network}

The input is an observation $\bm{y}$ and $M$ samples $\bm{z}^0_1, \dots,
\bm{z}^0_M$ of $\rho(\cdot \mid \bm{y})$. A stack of permutation-equivariant
self-attention blocks~\citep{zaheer2017deepsets,lee2019settransformer} couples
the samples and emits one displacement per sample, so the nodes are
$\bm{z}_j = \bm{z}^0_j + \Delta_{\theta}(\bm{z}^0, \bm{y}, M)_j$ and permuting
the samples permutes the nodes. The observation enters through tokens from a
small encoder of $\bm{y}$, to which every sample cross-attends. The output
layer is zero-initialized, so an untrained network leaves the samples in
place.

No layer's shape depends on $M$. The node count enters as an input, through a
feature-wise modulation~\citep{perez2018film} of every block by an embedding
of $\log M$ and of the log spread of the sample set, so one network
serves any $M$.

\subsection{Closed-form weights}
\label{sec:weights}

Nodes alone are not a quadrature. With $\bm{K}(\bm{z})$ the Gram matrix of
the nodes and $\bm{\mu}(\bm{z})$ the kernel mean at them, the weights are the
minimizer of the ridged discrepancy \eqref{eq:mmd} over unit-sum vectors,
\begin{equation}
\label{eq:weights}
\bm{w}(\bm{z}) = \operatorname*{arg\,min}_{\bm{1}^{\top}\bm{w} = 1}\;
\bm{w}^{\top} \big(\bm{K}(\bm{z}) + \varepsilon \bm{I}\big) \bm{w}
- 2\, \bm{w}^{\top} \bm{\mu}(\bm{z}) ,
\end{equation}
which two linear solves give in closed form, and which training
differentiates through. The weights are signed, and they sum to one, so the
quadrature integrates constants exactly. The solved vector is used as
\eqref{eq:weights} returns it, a minimizer over a set that contains equal
weights, which is what the comparison with the samples rests on. A vector
normalized after an unconstrained solve is not, and offers no such guarantee.

\subsection{Training}
\label{sec:training}

The objective is the expected discrepancy \eqref{eq:mmd} of the moved
quadrature to $\rho(\cdot \mid \bm{y})$,
\begin{equation}
\label{eq:loss}
\begin{aligned}
\mathcal{L}(\theta) &= \E_{\bm{y},\, M,\, \bm{z}^0}
  \big[\, \bm{w}^{\top} \bm{K}(\bm{z})\, \bm{w}
  - 2\, \bm{w}^{\top} \hat{\bm{\mu}}(\bm{z}) \,\big] , \\[2pt]
\bm{z}_j &= \bm{z}^0_j + \big[\Delta_{\theta}(\bm{z}^0, \bm{y}, M)\big]_j ,
  \qquad \bm{w} = \bm{w}(\bm{z}) .
\end{aligned}
\end{equation}
Here $\Delta_{\theta}$ is the field of Section~\ref{sec:network}, which reads the
whole sample set with the observation and returns one displacement per sample,
and $\bm{w}(\bm{z})$ are the weights \eqref{eq:weights} solved at the
displaced nodes. The observation $\bm{y}$ is drawn from the family, the node
count $M$ is uniform on the deployed grid, and $\bm{z}^0$ are $M$ independent
samples of $\rho(\cdot \mid \bm{y})$, all redrawn at each step. The gradient passes through the weight solve, and
the self-affinity is dropped because it does not depend on $\theta$.
The kernel mean $\hat{\bm{\mu}}$ is closed form where
Section~\ref{sec:worstcase} provides it, and otherwise an estimate from samples of
the reference at that observation. Training requires joint samples of the family,
and there are no labels and no optimization per observation. Sizes, optimizer and compute are in
Appendix~\ref{app:setup}.

\subsection{Comparison with the samples}
\label{sec:selection}

A trained network meets observations unlike any it saw, and nothing so far
stops it from returning a quadrature worse than the samples it was handed.
We therefore score three candidates by the same criterion and return the
best of them, nodes and weights together
(Figure~\ref{fig:architecture}). They are the samples $Q_0$ under equal
weights, the reweighted samples $Q_{\mathrm{rw}}$ under the weights
\eqref{eq:weights}, and the moved quadrature $Q_{\mathrm{mv}}$, the displaced
nodes under the weights solved there. The criterion $\widehat{J}$ is the
integrand of \eqref{eq:loss}, so the field is selected by the quantity it was
trained on, with the kernel mean exact or estimated. The samples never strictly win it
against their own reweighting, whatever the estimate, since
$Q_{\mathrm{rw}}$ minimizes the ridged criterion at the same nodes, equal
weights are feasible there, and no unit-sum vector has a smaller squared
norm than they do. The live comparison is between the two solved
candidates, and the samples enter it as a check. The criterion ranks by the
discrepancy to the reference itself, or by an estimate of it whose error the
next section bounds.
\section{Never-worse guarantee}
\label{sec:guarantee}

Two things must hold for Section~\ref{sec:field} to be more than a construction:
that solving the weights gains anything over the samples, and that the
comparison of Section~\ref{sec:selection} ranks by the true error at moved nodes and
at observations the network never saw. The first is a positive statement with
an explicit constant. The second is exact when the kernel mean is exact and
holds within a computable slack when it is estimated. Throughout, $k$ is
positive semidefinite with $k(\bm{x},\bm{x}) = 1$, the ridge is
$\varepsilon > 0$, the candidates have unit-sum weights and include the
samples, the network is any measurable map, and the samples are $M$
independent samples of $\rho$. Appendix~\ref{app:proofs} states and proves
everything in full, and the formal results are machine-checked in Lean~4.

\subsection{The reweighting gain}
\label{sec:gain}

\begin{proposition}[Strict gain of reweighting]
\label{prop:gain}
Let $k \ge 0$ pointwise, $c_\rho < 1$, $M \ge 3$, and let the kernel mean be
exact. Then
\begin{equation}
\label{eq:gain}
\begin{aligned}
&\mathbb{E}\big[\mathrm{MMD}^2(Q_0, \rho) - \mathrm{MMD}^2(Q_{\mathrm{rw}}, \rho)\big] \\
&\qquad\ge\; \frac{M - 2}{(1 + \varepsilon)\, M}\; r_\rho\; \frac{1 - c_\rho}{M},
\end{aligned}
\end{equation}
where $r_\rho = \mathbb{E}\,\bar k(X, X')^2 / (1 - c_\rho) \in (0, 1 - c_\rho]$
for independent $X, X' \sim \rho$ and $\bar k(\bm{x}, \bm{x}') = k(\bm{x}, \bm{x}')
- \mu_\rho(\bm{x}) - \mu_\rho(\bm{x}') + c_\rho$ the centered kernel, so that
$r_\rho = \operatorname{tr}(\Sigma_\rho^2) / \operatorname{tr}\Sigma_\rho$ for
the covariance operator $\Sigma_\rho$ of $k(\cdot, \bm{x}) - \mu_\rho$ under
$\rho$. The dependence on $\rho$ cannot be removed: already for the kernel
$k(\bm{x}, \bm{x}') = \mathbf{1}\{\bm{x} = \bm{x}'\}$ on a countable set, no
positive constant
depending only on $M$ bounds the left side below by a fraction of the floor
for every reference.
\end{proposition}

The right side is strictly positive, so reweighting the reference's own
samples lowers the expected error below the floor by a computable fraction of
it, at every budget from three. The bound is a strict positivity at each $M$,
and its factor tends to a constant. The measured gain is far larger
(Section~\ref{sec:depth}). The target enters through $r_\rho$, and by the second
half of the statement it must, since no constant in $M$ and $k$ alone bounds
the gain. The ratio $r_\rho$ is large when the centered spectrum of $\rho$
concentrates on few modes and collapses when $\rho$ spreads its mass over
many kernel-orthogonal directions. The proposition is stated for the exact kernel mean, and
Theorem~\ref{thm:neverworse}(ii) covers the comparison when it is estimated.

\subsection{Never worse than the samples}
\label{sec:neverworse}

The inequality of \citet[eq.~3]{belhadji2026datadriven} compares weightings
at fixed nodes. The comparison here ranges over configurations whose nodes
have moved, for the constrained solve \eqref{eq:weights}, at observations
and node counts the network never saw, and for any network. The mechanism is
that, with the exact kernel mean, $\widehat{J}$ equals the discrepancy up
to $c_\rho$, and that the samples are a candidate.

\begin{theorem}[Never worse than the samples]
\label{thm:neverworse}
Let $\hat a$ be the candidate with the smallest $\widehat{J}$ among $Q_0$,
$Q_{\mathrm{rw}}$ and $Q_{\mathrm{mv}}$ (Section~\ref{sec:selection}).
\begin{enumerate}
\item[(i)] If the kernel mean is exact, then at every observation and every
$M$,
$\mathrm{MMD}^2(Q_{\hat a}, \rho) = \min_a \mathrm{MMD}^2(Q_a, \rho)
\le \mathrm{MMD}^2(Q_0, \rho)$.
\item[(ii)] Let the comparison be scored on $L_2 \ge 2$ independent samples of
$\rho$, drawn independently of the observation, the input samples, the dataset
used to solve the weights \eqref{eq:weights}, and the network. For two candidates $a, b$, let
$D_{ab} = \mathrm{MMD}(Q_a, Q_b)$, let $\widehat{V}_{ab}$ be the sample
variance over these $L_2$ samples of $h_a - h_b$, where
$h_a(\bm{x}) = \sum_j w_j k(\bm{z}_j, \bm{x})$ for $Q_a$, and let
\begin{equation}
\label{eq:slack}
\hat\gamma_{ab}(\delta)
= 2 \sqrt{\frac{2 \widehat{V}_{ab} \ln(4/\delta)}{L_2}}
+ \frac{28}{3}\, \frac{D_{ab} \ln(4/\delta)}{L_2 - 1} .
\end{equation}
Then on one event of probability at least $1 - \delta$,
$\mathrm{MMD}^2(Q_{\hat a}, \rho) \le \mathrm{MMD}^2(Q_0, \rho)
+ \hat\gamma_{\hat a, 0}(\delta/3)$, and between the two solved candidates
the choice is exact whenever the observed gap of $\widehat{J}$ exceeds
$\hat\gamma_{\mathrm{rw},\mathrm{mv}}(\delta/3)$, and otherwise the excess
over the better of them is at most that slack.
\end{enumerate}
\end{theorem}

Part~(i) requires no hypothesis on the network and none on the kernel other than positive semidefiniteness, and it says nothing about the moved candidate
alone. On a trained reference the kernel mean is estimated,
the ranking is noisy, and the comparison can return the worse candidate.
Part~(ii) is the working guarantee. The slack \eqref{eq:slack} is computed
from the split alone and scales with the distance between the two candidates
compared, because sampling noise common to both cancels in the paired
difference. Its form is that of the empirical Bernstein inequality, and with
Hoeffding's inequality in its place the slack becomes
$2 D_{ab} \sqrt{2 \ln(6/\delta) / L_2}$, which does not shrink with the
variance (Corollary~\ref{cor:hoeffding}). For an untrained network the two solved candidates coincide, so
part~(i) reduces to the reweighting inequality and in part~(ii) their slack
is identically zero, while the slack against the samples is not. We call a
comparison whose observed gap exceeds its slack \emph{certified}, and the
$L_2$ samples the \emph{certification split}. The split is drawn afresh, so
certification is paid in reference samples alone and the candidates are
unchanged, and what it certifies is the returned quadrature against the
reference as supplied (Remark~\ref{rem:reference}). We run part~(ii) at
every returned quadrature of the groundwater problem in Section~\ref{sec:certification} and
measure in Section~\ref{sec:depth} what Proposition~\ref{prop:gain} bounds.
\section{Numerical experiments}
\label{sec:experiments}

The experiments measure the margin below the floor, which the guarantees
leave open, and run the two statements of Section~\ref{sec:guarantee} where
they apply. The references are, in the order they are harder to obtain, a
two-dimensional Gaussian and a ten-dimensional Gaussian mixture in closed
form, a two-dimensional banana whose reference is a dataset of samples, and
the exact conditional posterior of limited-angle tomography. Three trained
conditional flows follow, on limited-angle tomography at a finer grid, on
groundwater flow and on a Lotka--Volterra family. Each is specified in
Appendix~\ref{app:setup}. The figures show the discrepancy \eqref{eq:mmd} relative
to the floor \eqref{eq:floor} as a median over held-out observations or
seed sets, at the node counts $M \in \{32, 64, 128, 256, 512\}$ the field
was trained on, with bands over training seeds where they exist. We call
one reference, functional, or member at one node count a \emph{cell}.

\begin{figure*}[t]
\centering
\includegraphics[width=\textwidth]{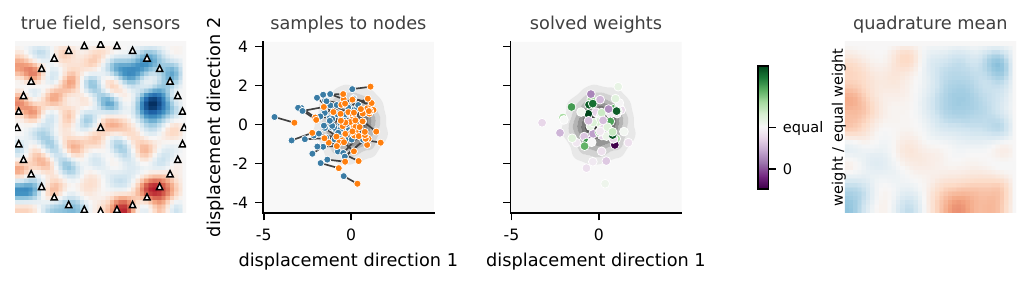}
\caption{One forward pass on a held-out groundwater observation: the true
field and sensors, the samples joined to their nodes, the solved weights, and
the posterior mean under the returned quadrature.}
\label{fig:field}
\end{figure*}

\subsection{Error against the node count}
\label{sec:depth}

The tomography posterior is exact in a rank-seven informed subspace of a
$16 \times 16$ image, and the trained flow is read in a rank-32 subspace of
a $64 \times 64$ image. The groundwater problem has a hundred-dimensional
permeability field observed at 33 sensors, with a conditional
flow~\citep{kruse2021hint} as reference. We verify the groundwater flow at
$M = 64$ against preconditioned Crank--Nicolson chains that solve the
forward problem exactly. The flow's distance to the chains is at the
chains' own reproducibility, so the trained reference approximates the
posterior as accurately as the chains resolve (Appendix~\ref{app:gw}).
The returned quadrature is below the floor against the chains on every
observation.
Figure~\ref{fig:field} shows one forward pass on a held-out groundwater
observation and the posterior mean the returned quadrature yields.

Figure~\ref{fig:depth} shows the returned quadrature below the floor at
each node count on the six references it plots. The margin widens with $M$
on five of the six and flattens on the trained tomography flow. Where the
node Gram stays well conditioned the quadrature reaches three orders of
magnitude below the floor, on the ten-dimensional mixture at the largest
node count. The two low-dimensional closed forms run deeper but are
ill-conditioned above the smallest budget, so those cells are drawn open
and read as values. The two trained flows hold a shallower margin below the floor at
every node count.

Two controls on the groundwater problem isolate the displacement's use of
the observation. A network fed a shuffled observation lands clearly above
the returned quadrature, one fed a constant observation closer but still
above, and both stay well below the floor. The displacement's share of the
log margin is largest at the smallest node count and shrinks to a small
fraction at the largest, on the references readable there. On one mixture the moved
candidate is a net loss, which the comparison of
Section~\ref{sec:selection} covers.

\begin{figure}[t]
\centering
\includegraphics[width=\columnwidth]{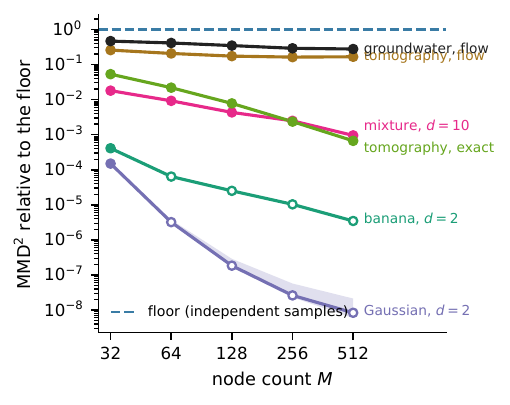}
\caption{Discrepancy against $M$. Open markers are ill-conditioned cells and
bands span training seeds.}
\label{fig:depth}
\end{figure}

\subsection{The certification split}
\label{sec:certification}

On the references whose kernel mean is exact, a sweep over budgets from 16
to 2048, two hundred thousand quadratures in all, returned none worse than
its samples, as Theorem~\ref{thm:neverworse}(i) requires. This checks the implementation, not the statement. On the groundwater
problem the kernel mean is estimated, and we ran
Theorem~\ref{thm:neverworse}(ii) at each returned quadrature: 24
observations, eight redraws of the samples, five node counts and eleven
split sizes. The split is drawn afresh from the flow under a seed rule
independent of the data, at a family error of $0.05$ charged as $\delta/3$
per pair.

Figure~\ref{fig:certification} shows the certified fraction against the
split size. Across more than twenty thousand comparisons, the samples are
never certified better than either solved candidate. Nothing certifies
below a thousand split samples at any node count, and from a few thousand
every comparison is certified, fewer split samples than the problem
already spends on the kernel mean. In a small fraction of the comparisons between
the two solved candidates the moved one is certified worse---the
comparison at work. The fractions are of marginal statements, one
quadrature against one split, and the simultaneous form moves the
all-certified threshold at most one grid step up.

\begin{figure}[t]
\centering
\includegraphics[width=\columnwidth]{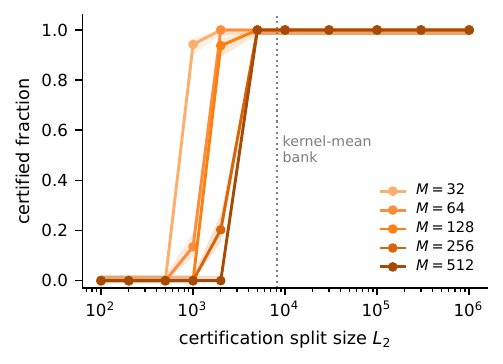}
\caption{Certified fraction against split size on the groundwater problem,
one curve per node count.}
\label{fig:certification}
\end{figure}

\begin{figure*}[t]
\centering
\includegraphics[width=\textwidth]{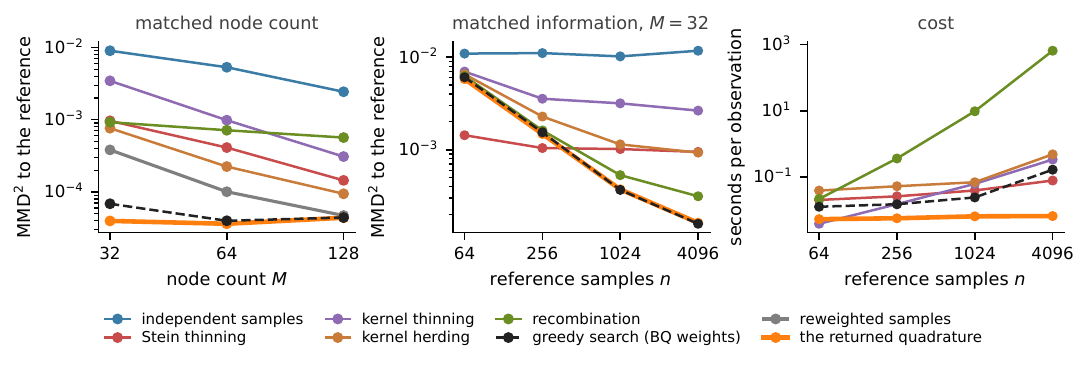}
\caption{Lotka--Volterra benchmark: discrepancy at matched node count (left)
and matched information (middle), cost (right), medians over the resolved
members.}
\label{fig:peers}
\end{figure*}

\subsection{Accuracy against forward solves}
\label{sec:solves}

Here the integrand is expensive---the motivation of Stein
thinning~\citep{riabiz2022steinthinning}, with a different object. Each
node of the groundwater quadrature is one solve of the pressure equation.
We integrate three functionals of the solved field: the pressure at an
unmetered probe, a smoothed tail exceedance registered in advance as
primary, and the flux across the midline. We use eight observations and
256 repeats per cell, a gold reference of 131072 solves per observation,
and a Holm correction across the family.

Table~\ref{tab:solves} shows the number of independent samples that match
the returned quadrature's accuracy at each budget. It is more accurate at
the same number of solves in all but one cell, the flux at $M = 64$, where
the gold reference is not resolved finely enough to separate the
candidates. Reweighting alone provides essentially the whole multiplier at
the largest budget, where the moved candidate no longer separates from it.
A reader who charges the quadrature for its input samples as well as for
its nodes compares it with twice the independent solves. That comparison
splits. The quadrature wins at the larger budgets, and below them it wins
only on the smooth probe pressure, losing the rest of the resolved cells,
the primary functional among them.
The instrument's own control, independent sampling at twice the budget,
recovers the expected factor $\sqrt{2}$ in each cell.

\begin{table}[t]
\caption{Independent samples matching the returned quadrature: stars extrapolated, dash
unresolved, bold above $2M$.}
\label{tab:solves}
\centering
\small
\begin{tabular}{lrrrrr}
\toprule
& \multicolumn{5}{c}{node count $M$} \\
\cmidrule(lr){2-6}
functional & 32 & 64 & 128 & 256 & 512 \\
\midrule
probe pressure & 56 & \textbf{132} & \textbf{406} & \textbf{1603}$^*$ & \textbf{6157}$^*$ \\
tail exceedance & 44 & 94 & 216 & \textbf{615}$^*$ & \textbf{1491}$^*$ \\
midline flux & 42 & -- & 225 & \textbf{717}$^*$ & \textbf{2491}$^*$ \\
\bottomrule
\end{tabular}
\end{table}

\subsection{Comparison with per-observation constructions}
\label{sec:peers}

The Lotka--Volterra benchmark of kernel and Stein
thinning~\citep{riabiz2022steinthinning} provides the comparison with
constructions that run once per target. We take eight held-out members of a
family whose posterior widths span more than two orders of magnitude, one at
the benchmark's generating parameter, with a conditional flow as reference.
The peers---kernel thinning, Stein thinning on the squared-exponential
Stein kernel, kernel herding, recombination, and a greedy search under
Bayesian-quadrature weights---each run through our implementation of the
published construction. Stein thinning is handed the flow's exact score by
automatic differentiation, while the field sees samples only. At matched
node count each peer emits its own output, scored on a held-out dataset
of 65536 samples whose own sampling error sets the resolution below which
candidates are not separated. At matched information, at $M = 32$, each method
receives the same $n$ reference samples, $n$ from 64 to 4096, on an earlier
training run of the field.

Figure~\ref{fig:peers}, left, shows the returned quadrature ahead of all
four compression peers in the resolved cells, at the three node counts they
run, by margins that reach nearly two orders of magnitude at the smallest. Reweighting alone is
ahead of all four. The quadrature is well ahead of the greedy search at the
smallest node count and level with it at the two larger ones. The one cell
that favors the search there is unresolved. In the middle, at matched
information, the quadrature is level with the greedy search at each $n$ and
ahead of the four compression peers once $n$ reaches a thousand, below
which Stein thinning, reading the exact score, is ahead. The quadrature's
time per observation stays at a few milliseconds while the pool grows
64-fold, and each peer's grows with it, recombination's to minutes.

That dataset resolves the margin on four of eight members at $M = 32$, and
on one member at larger $M$. In the resolved cells the quadrature is more
than two orders of magnitude below the floor, reweighting alone providing
more than an order of magnitude of it and the displacement the rest. The
flow is not the exact posterior, so this benchmark measures integration
of the reference (Remark~\ref{rem:reference}, Appendix~\ref{app:lv-family}).

The expensive integrand of Section~\ref{sec:solves} has a counterpart here,
one solve of the predator--prey system per node. Three functionals of the
trajectory---the time of the first prey peak, a smoothed exceedance of the
windowed peak, and the prey population past the data window---are
integrated against the flow (Appendix~\ref{app:lv-family}). In the paired
median the returned quadrature beats the same number of independent solves
in each cell. At the smallest node count, on the exceedance, reweighting
alone matches hundreds of independent solves on seven of eight members at
no extra cost. On the peak time the quadrature matches at least ten
thousand, since its median error lies at the gold reference's own
resolution. On the broadest member reweighting is below par, and the prey
population and the exceedance at the largest node count run the other way.
\section{Discussion}
\label{sec:discussion}

The guarantee covers the returned quadrature against the reference as
supplied, in the kernel's norm, at each observation and node count, for any
network when the kernel mean is exact, and within a computable slack when it
is estimated. We prove no rate, and the experiments show where one is
missing: the margin flattens on the trained tomography flow, and on one
mixture the moved candidate is a net loss that the comparison absorbs.
Outside the kernel's ball the same signed weights that buy the margin can
cost accuracy, and on a discontinuous indicator the returned quadrature can
integrate worse than the samples do. The certificate itself is paid for in
samples, at the split sizes of Section~\ref{sec:certification}, so a
reference that is costly to sample earns a weaker certificate.

The kernel decides which errors the construction is minimizing, so
choosing or learning it would move the guarantee from the norm we fix
toward the integral a practitioner reports. Stein kernels would extend the
field to targets known only through their score. The comparison is not
closed---any candidate with unit-sum weights can enter it, including the
compressions (Section~\ref{sec:peers}).
\section{Conclusion}
\label{sec:conclusion}

We presented the quadrature field: one network that returns, at any
observation and any node count, a weighted quadrature of the reference
posterior in one forward pass. Existing designed quadratures are built
for one observation, and the next observation pays for its own; the
field pays once, at training. The returned quadrature is the best of
three candidates, the samples among them, so it is never worse than the
samples it was handed. That holds without qualification when the kernel
mean is exact, and with high probability within a computable slack when
it is estimated. Reweighting alone already lowers the expected error
below the floor by a fraction set by the reference's own spectrum.
Across six references the returned quadrature is below the floor in
median at every node count. Where the reference spreads its mass over
many kernel-orthogonal directions, that fraction collapses toward zero,
and we expect the field to fall silent there.

\section*{AI Use Statement}

Generative AI tools were used to aid polishing the writing of this paper,
to assist in providing Lean code for machine-checking its mathematical
proofs, and to partially implement code. All AI-assisted work was reviewed
by the authors, who take responsibility for the final content of this
work, including text and code.

\bibliographystyle{apalike}
\bibliography{refs}

\clearpage
\appendix
\onecolumn
\setcounter{theorem}{0}
\renewcommand{\thetheorem}{\thesection.\arabic{theorem}}
\setcounter{proposition}{0}
\renewcommand{\theproposition}{\thesection.\arabic{proposition}}
\section{Proofs}
\label{app:proofs}

Every statement of Section~\ref{sec:quadrature}--Section~\ref{sec:guarantee} is stated
here in full and proved. The formal results are machine-checked in Lean~4.

\subsection{The setting}
\label{app:setting}

Let $\mathcal{H}$ be a real Hilbert space and $\Phi : \mathcal{X} \to
\mathcal{H}$ a feature map, so that $k(\bm{x}, \bm{x}') = \ip{\Phi(\bm{x})}{\Phi(\bm{x}')}$
is a positive semidefinite kernel; every positive semidefinite kernel arises
this way, with $\mathcal{H}$ its reproducing kernel Hilbert space and
$\Phi(\bm{x}) = k(\cdot, \bm{x})$. A quadrature with $M$ nodes is a pair
$Q = (\bm{z}, \bm{w})$ of nodes $\bm{z} = (\bm{z}_1, \dots, \bm{z}_M)$,
coincident entries allowed, and signed weights $\bm{w} \in \mathbb{R}^M$, with
embedding $\Phi_Q = \sum_j w_j \Phi(\bm{z}_j)$. For $\bm{m} \in \mathcal{H}$,
\begin{equation}
\label{eq:app-mmd}
\begin{gathered}
\MMD^2(Q, \bm{m}) = \|\Phi_Q - \bm{m}\|^2, \\
\bm{K}(\bm{z})_{ij} = k(\bm{z}_i, \bm{z}_j), \qquad
\bm{\mu}(\bm{z})_j = \ip{\Phi(\bm{z}_j)}{\bm{m}},
\end{gathered}
\end{equation}
and, for any $\hat{\bm{\mu}} \in \mathbb{R}^M$ in place of $\bm{\mu}(\bm{z})$,
$\widehat J(\bm{w}) = \bm{w}^{\!\top} \bm{K}(\bm{z}) \bm{w} - 2 \bm{w}^{\!\top} \hat{\bm{\mu}}$.
Unit sum means $\sum_j w_j = 1$, and $\bm{w}_{\mathrm{eq}} = \tfrac1M \bm{1}$.
When a probability measure $\rho$ enters, $\Phi$ is strongly measurable and
bounded, $\bm{m} = \int \Phi \, \mathrm{d}\rho$ is the Bochner integral, and
\begin{equation}
\label{eq:app-mean}
\mu_\rho(\bm{x}) = \int k(\bm{x}, \bm{x}') \, \mathrm{d}\rho(\bm{x}') = \ip{\Phi(\bm{x})}{\bm{m}},
\qquad
c_\rho = \iint k \, \mathrm{d}\rho \, \mathrm{d}\rho = \|\bm{m}\|^2.
\end{equation}

\begin{lemma}[Expansion of the discrepancy]
\label{lem:app-expand}
For every quadrature and every $\bm{m}$,
\[
\MMD^2(Q, \bm{m}) = \sum_{i,j} w_i w_j k(\bm{z}_i, \bm{z}_j) - 2 \sum_j w_j \ip{\Phi(\bm{z}_j)}{\bm{m}} + \|\bm{m}\|^2 .
\]
\end{lemma}

\begin{proof}
\[
\|\Phi_Q - \bm{m}\|^2 = \|\Phi_Q\|^2 - 2\ip{\Phi_Q}{\bm{m}} + \|\bm{m}\|^2,
\quad
\|\Phi_Q\|^2 = \sum_{i,j} w_i w_j \ip{\Phi(\bm{z}_i)}{\Phi(\bm{z}_j)},
\quad
\ip{\Phi_Q}{\bm{m}} = \sum_j w_j \ip{\Phi(\bm{z}_j)}{\bm{m}} .
\qedhere
\]
\end{proof}

\begin{lemma}[The worst case over the unit ball]
\label{lem:app-worst}
For every quadrature and every $\bm{m}$, the set
$\{ |\sum_j w_j \ip{f}{\Phi(\bm{z}_j)} - \ip{f}{\bm{m}}| : \|f\| \le 1 \}$ has
greatest element $\MMD(Q, \bm{m})$.
\end{lemma}

\begin{proof}
\[
\Big| \sum_j w_j \ip{f}{\Phi(\bm{z}_j)} - \ip{f}{\bm{m}} \Big|
= |\ip{f}{\Phi_Q - \bm{m}}| \le \|f\| \, \|\Phi_Q - \bm{m}\| \le \|\Phi_Q - \bm{m}\|,
\]
with equality at $f = (\Phi_Q - \bm{m})/\|\Phi_Q - \bm{m}\|$, or at $f = 0$
when $\Phi_Q = \bm{m}$. For the canonical feature map,
$\ip{f}{\Phi(\bm{x})} = f(\bm{x})$ and $\ip{f}{\bm{m}} = \int f \, \mathrm{d}\rho$
by \eqref{eq:app-mean}, which is the reading of Section~\ref{sec:worstcase}.
\end{proof}

\begin{lemma}[The self-affinity cancels]
\label{lem:app-crit}
Under the exact estimand $\hat{\bm{\mu}} = \bm{\mu}(\bm{z})$, for every
node configuration and every weight vector,
$\widehat J(\bm{w}) = \MMD^2(Q, \bm{m}) - \|\bm{m}\|^2$.
\end{lemma}

\begin{proof}
Lemma~\ref{lem:app-expand} with \eqref{eq:app-mmd}.
\end{proof}

\begin{proof}[Proof of Remark~\ref{rem:reference}]
With $\bm{m}_\rho$ and $\bm{m}_\pi$ the two kernel means,
$\|\Phi_Q - \bm{m}_\pi\| \le \|\Phi_Q - \bm{m}_\rho\| + \|\bm{m}_\rho - \bm{m}_\pi\|$.
\end{proof}

\subsection{The floor identity}
\label{app:floor}

\begin{proposition}[The floor, general form]
\label{prop:floor-general}
Let $\Phi$ be strongly measurable with $\|\Phi\| \le B$, let $\rho$ be a
probability measure, and let $\xi_1, \dots, \xi_M$, $M \ge 1$, be independent
with law $\rho$. For $Q_0 = ((\xi_j)_j, \bm{w}_{\mathrm{eq}})$ and $X \sim \rho$,
\[
\E\, \MMD^2(Q_0, \bm{m}) = \frac{\E\, k(X, X) - c_\rho}{M} .
\]
When $k(\bm{x},\bm{x}) = 1$ this is display~\eqref{eq:floor},
and in general
$c_\rho \le \E\, k(X,X)$.
\end{proposition}

\begin{proof}
For $i \ne j$ the law of $(\xi_i, \xi_j)$ is $\rho \otimes \rho$ and the inner
product commutes with the Bochner integral, so
\[
\E \ip{\Phi(\xi_i)}{\Phi(\xi_j)} = \iint \ip{\Phi(\bm{x})}{\Phi(\bm{x}')} \, \mathrm{d}\rho(\bm{x}') \, \mathrm{d}\rho(\bm{x}) = \|\bm{m}\|^2,
\qquad
\E \|\Phi(\xi_i)\|^2 = \E\, k(X,X),
\]
\[
\E \ip{\Phi(\xi_j)}{\bm{m}} = \|\bm{m}\|^2.
\]
By
Lemma~\ref{lem:app-expand} with equal weights,
\[
\E\, \MMD^2(Q_0, \bm{m})
= \frac{M\, \E k(X,X) + M(M-1)\, c_\rho}{M^2} - 2 c_\rho + c_\rho
= \frac{\E\, k(X,X) - c_\rho}{M},
\]
\[
c_\rho = \|\bm{m}\|^2 \le \Big( \int \|\Phi\| \, \mathrm{d}\rho \Big)^2 \le \int \|\Phi\|^2 \, \mathrm{d}\rho,
\]
the last step from $\int (\|\Phi\| - a)^2 \, \mathrm{d}\rho \ge 0$ at
$a = \int \|\Phi\| \, \mathrm{d}\rho$.
\end{proof}

\subsection{The weights and the comparison with the samples}
\label{app:weights}

Fix a symmetric $\bm{A} \in \mathbb{R}^{M \times M}$ and $\hat{\bm{\mu}} \in \mathbb{R}^M$,
and write $F(\bm{v}) = \bm{v}^{\!\top} \bm{A} \bm{v} - 2 \bm{v}^{\!\top} \hat{\bm{\mu}}$;
the ridged criterion of Section~\ref{sec:weights} is $F$ with $\bm{A} = \bm{K}(\bm{z}) + \varepsilon \bm{I}$,
and a \emph{ridged constrained minimizer} is a unit-sum $\bm{w}$ with
$F(\bm{w}) \le F(\bm{v})$ for every unit-sum $\bm{v}$.

\begin{lemma}[Sufficiency of the first-order condition]
\label{lem:app-kkt}
Let $\bm{u}^{\!\top} \bm{A} \bm{u} \ge 0$ for all $\bm{u}$, let $\bm{w}$ have
unit sum, and let $\bm{A}\bm{w} - \hat{\bm{\mu}} = c\,\bm{1}$ for some
$c \in \mathbb{R}$. Then $F(\bm{w}) \le F(\bm{v})$ for every unit-sum $\bm{v}$.
\end{lemma}

\begin{proof}
For unit-sum $\bm{v}$,
\[
F(\bm{v}) - F(\bm{w})
= (\bm{v} - \bm{w})^{\!\top} \bm{A} (\bm{v} - \bm{w}) + 2 (\bm{v} - \bm{w})^{\!\top} (\bm{A}\bm{w} - \hat{\bm{\mu}})
= (\bm{v} - \bm{w})^{\!\top} \bm{A} (\bm{v} - \bm{w}) + 2c \sum_j (v_j - w_j)
\ge 0 .
\qedhere
\]
\end{proof}

\begin{proposition}[The closed-form weights]
\label{prop:app-closedform}
Let $\bm{u}^{\!\top} \bm{A} \bm{u} > 0$ for all $\bm{u} \ne 0$. Then $\bm{A}$
is invertible, $\bm{1}^{\!\top} \bm{A}^{-1} \bm{1} > 0$, and
\[
\bm{w}^\star = \bm{A}^{-1} \hat{\bm{\mu}} - \lambda\, \bm{A}^{-1} \bm{1},
\qquad
\lambda = \frac{\bm{1}^{\!\top} \bm{A}^{-1} \hat{\bm{\mu}} - 1}{\bm{1}^{\!\top} \bm{A}^{-1} \bm{1}},
\]
is the unique minimizer of $F$ over unit-sum vectors. At every node
configuration and every estimand, $\bm{A} = \bm{K}(\bm{z}) + \varepsilon \bm{I}$
with $\varepsilon > 0$ qualifies.
\end{proposition}

\begin{proof}
\[
\bm{A}\bm{u} = \bm{A}\bm{v} \;\Rightarrow\; (\bm{u}-\bm{v})^{\!\top}\bm{A}(\bm{u}-\bm{v}) = 0 \;\Rightarrow\; \bm{u} = \bm{v};
\qquad
\bm{v}^{\!\top} \bm{A}^{-1} \bm{v} = (\bm{A}^{-1}\bm{v})^{\!\top} \bm{A} (\bm{A}^{-1}\bm{v}) > 0 \quad (\bm{v} \ne 0);
\]
\[
\bm{1}^{\!\top} \bm{w}^\star = \bm{1}^{\!\top} \bm{A}^{-1} \hat{\bm{\mu}} - \lambda\, \bm{1}^{\!\top} \bm{A}^{-1} \bm{1} = 1,
\qquad
\bm{A} \bm{w}^\star - \hat{\bm{\mu}} = -\lambda \bm{1},
\]
so Lemma~\ref{lem:app-kkt} applies; for a second unit-sum minimizer
$\bm{w}'$, the expansion in its proof gives
$0 = F(\bm{w}') - F(\bm{w}^\star) = (\bm{w}' - \bm{w}^\star)^{\!\top} \bm{A} (\bm{w}' - \bm{w}^\star)$,
so $\bm{w}' = \bm{w}^\star$. Finally
$\bm{u}^{\!\top} (\bm{K}(\bm{z}) + \varepsilon\bm{I}) \bm{u} = \|\sum_j u_j \Phi(\bm{z}_j)\|^2 + \varepsilon \|\bm{u}\|^2 > 0$
for $\bm{u} \ne 0$.
\end{proof}

\begin{proposition}[The samples never strictly win the criterion]
\label{prop:app-twoarm}
Let $\bm{K}$ be any matrix, $\hat{\bm{\mu}}$ any vector, $\varepsilon \ge 0$,
and $\bm{w}_{\mathrm{rw}}$ a ridged constrained minimizer. Then
\[
\widehat J(\bm{w}_{\mathrm{rw}})
\;\le\; \widehat J(\bm{w}_{\mathrm{eq}}) - \varepsilon\big(\|\bm{w}_{\mathrm{rw}}\|^2 - \tfrac1M\big)
\;\le\; \widehat J(\bm{w}_{\mathrm{eq}}) .
\]
\end{proposition}

\begin{proof}
$\bm{w}_{\mathrm{eq}}$ has unit sum and $\|\bm{w}_{\mathrm{eq}}\|^2 = 1/M$, so
\[
\widehat J(\bm{w}_{\mathrm{rw}}) + \varepsilon \|\bm{w}_{\mathrm{rw}}\|^2 \le \widehat J(\bm{w}_{\mathrm{eq}}) + \frac{\varepsilon}{M},
\qquad
1 = \Big( \sum_j w_{\mathrm{rw},j} \Big)^2 \le M \|\bm{w}_{\mathrm{rw}}\|^2.
\]
\end{proof}

\begin{corollary}[The ridged chain in the discrepancy]
\label{cor:app-ridgedchain}
Under the exact estimand, with $Q_{\mathrm{rw}} = (\bm{z}, \bm{w}_{\mathrm{rw}})$
and $Q_0 = (\bm{z}, \bm{w}_{\mathrm{eq}})$,
\[
\MMD^2(Q_{\mathrm{rw}}, \bm{m}) \le \MMD^2(Q_0, \bm{m}) - \varepsilon\big(\|\bm{w}_{\mathrm{rw}}\|^2 - \tfrac1M\big) \le \MMD^2(Q_0, \bm{m}) .
\]
\end{corollary}

\begin{proof}
Proposition~\ref{prop:app-twoarm} and Lemma~\ref{lem:app-crit}.
\end{proof}

\subsection{Theorem~\ref{thm:neverworse}(i): the exact kernel mean}
\label{app:exact}

\begin{theorem}[Exact selection]
\label{thm:app-exact}
Let $k$ be positive semidefinite, let $\{Q_a\}_{a \in \mathcal{A}}$ be a
finite family of quadratures scored by $\widehat J$ under the exact
estimand, and let $\hat a$ minimize $\widehat J$ over $\mathcal{A}$. Then
$\MMD^2(Q_{\hat a}, \bm{m}) = \min_{a} \MMD^2(Q_a, \bm{m})$. For the three
candidates of Section~\ref{sec:selection}, with $Q_{\mathrm{rw}}$ a ridged
constrained minimizer at the samples,
\[
\MMD^2(Q_{\hat a}, \bm{m}) = \min_{a} \MMD^2(Q_a, \bm{m})
\le \MMD^2(Q_{\mathrm{rw}}, \bm{m})
\le \MMD^2(Q_0, \bm{m}) - \varepsilon\big(\|\bm{w}_{\mathrm{rw}}\|^2 - \tfrac1M\big)
\le \MMD^2(Q_0, \bm{m}) .
\]
\end{theorem}

\begin{proof}
By Lemma~\ref{lem:app-crit}, $\widehat J(Q_a) = \MMD^2(Q_a, \bm{m}) - \|\bm{m}\|^2$
for every $a$; the chain is Corollary~\ref{cor:app-ridgedchain}.
\end{proof}

The nodes may come from any map and may coincide; nothing is assumed of the
kernel beyond positive semidefiniteness. For an untrained network the moved
candidate equals the reweighted one.

\subsection{Proposition~\ref{prop:gain}: the reweighting gain}
\label{app:gain}

Throughout this subsection $k(\bm{x},\bm{x}) = 1$, $\rho$ is a probability
measure, $\bm{m} = \int \Phi \, \mathrm{d}\rho$, and the estimand is exact;
$\bar\Phi(\bm{x}) = \Phi(\bm{x}) - \bm{m}$ is the centered feature,
$\bar k(\bm{x}, \bm{x}') = \ip{\bar\Phi(\bm{x})}{\bar\Phi(\bm{x}')}$ the centered
kernel of Proposition~\ref{prop:gain}, and
\begin{equation}
\label{eq:app-rtimes}
\begin{aligned}
T_\rho &= \iint \bar k(\bm{x}, \bm{x}')^2 \, \mathrm{d}\rho(\bm{x}) \, \mathrm{d}\rho(\bm{x}')
= (1 - c_\rho)\, r_\rho , \\
V_\rho &= \int \ip{\bm{m}}{\bar\Phi(\bm{x})}^2 \, \mathrm{d}\rho(\bm{x})
= \operatorname{Var}_\rho(\mu_\rho) ,
\end{aligned}
\end{equation}
with $\|\bm{m}\| \le 1$, $\|\bar\Phi\| \le 2$ and $\int \bar\Phi \, \mathrm{d}\rho = 0$.

\begin{lemma}[The one-pair gain, for every realization]
\label{lem:app-pair}
Fix a configuration $\bm{z}$, $\varepsilon \ge 0$, a ridged constrained
minimizer $\bm{w}_\varepsilon$, and $i \ne j$; put
$D = \Phi_{Q_{\mathrm{eq}}} - \bm{m}$ and $\bm{u} = \Phi(\bm{z}_i) - \Phi(\bm{z}_j)$.
For every $b \ge \|\bm{u}\|^2$ with $b + 2\varepsilon > 0$,
\[
\frac{\ip{D}{\bm{u}}^2}{b + 2\varepsilon}
\;\le\; \MMD^2(Q_{\mathrm{eq}}, \bm{m}) - \MMD^2\big((\bm{z}, \bm{w}_\varepsilon), \bm{m}\big);
\]
$b = 2$ is admissible when $k \ge 0$, and $b = 4$ always.
\end{lemma}

\begin{proof}
For $t \in \mathbb{R}$ the vector $\bm{w}_t = \bm{w}_{\mathrm{eq}} + t(\bm{e}_i - \bm{e}_j)$
has unit sum, and
\[
\Phi_{(\bm{z}, \bm{w}_t)} = \Phi_{Q_{\mathrm{eq}}} + t\bm{u}, \qquad
\|\bm{w}_t\|^2 = \tfrac1M + 2t^2, \qquad
\MMD^2\big((\bm{z}, \bm{w}_t), \bm{m}\big) = \|D\|^2 + 2t\ip{D}{\bm{u}} + t^2\|\bm{u}\|^2 .
\]
Ridged minimality against $\bm{w}_t$ and $\|\bm{w}_\varepsilon\|^2 \ge 1/M$
(Proposition~\ref{prop:app-twoarm}) give, for every $t$,
\[
\MMD^2(Q_{\mathrm{eq}}, \bm{m}) - \MMD^2\big((\bm{z}, \bm{w}_\varepsilon), \bm{m}\big)
\ge -2t\ip{D}{\bm{u}} - t^2\big(\|\bm{u}\|^2 + 2\varepsilon\big)
\ge -2t\ip{D}{\bm{u}} - t^2 (b + 2\varepsilon),
\]
which equals $\ip{D}{\bm{u}}^2/(b + 2\varepsilon)$ at
$t = -\ip{D}{\bm{u}}/(b + 2\varepsilon)$;
and $\|\bm{u}\|^2 = 2 - 2k(\bm{z}_i, \bm{z}_j) \le 2$ when $k \ge 0$,
$\|\bm{u}\| \le \|\Phi(\bm{z}_i)\| + \|\Phi(\bm{z}_j)\| = 2$ always.
\end{proof}

\begin{lemma}[The pair inner product, decomposed]
\label{lem:app-decomp}
Let $M = n + 2$, $\bm{z} = (\bm{x}_0, \bm{x}_1, \bm{y}_1, \dots, \bm{y}_n)$,
$\bm{u} = \Phi(\bm{x}_0) - \Phi(\bm{x}_1)$, $\beta(\bm{y}) = \ip{\bar\Phi(\bm{y})}{\bm{u}}$
and $\Delta = \ip{\bm{m}}{\bm{u}}$. Then
$\ip{D}{\bm{u}} = \frac1M \sum_{l=1}^{n} \beta(\bm{y}_l) - \frac2M \Delta$.
\end{lemma}

\begin{proof}
\[
\ip{D}{\bm{u}}
= \frac1M \Big( \ip{\Phi(\bm{x}_0)}{\bm{u}} + \ip{\Phi(\bm{x}_1)}{\bm{u}} + \sum_{l} \ip{\Phi(\bm{y}_l)}{\bm{u}} \Big) - \Delta
= \frac1M \sum_l \ip{\bar\Phi(\bm{y}_l)}{\bm{u}} + \frac{n}{M}\Delta - \Delta,
\]
since $\ip{\Phi(\bm{x}_0)}{\bm{u}} + \ip{\Phi(\bm{x}_1)}{\bm{u}} = \|\Phi(\bm{x}_0)\|^2 - \|\Phi(\bm{x}_1)\|^2 = 0$
and $n/M - 1 = -2/M$.
\end{proof}

\begin{lemma}[The conditional second moment]
\label{lem:app-condmoment}
Let $\beta : \mathcal{X} \to \mathbb{R}$ be bounded and measurable with
$\int \beta \, \mathrm{d}\rho = 0$, let $\bm{y}_1, \dots, \bm{y}_n$ be
independent with law $\rho$, and let $a \ne 0$, $\Delta \in \mathbb{R}$. Then
\[
\E \Big[ \Big( \frac1a \sum_{l=1}^n \beta(\bm{y}_l) - \frac2a \Delta \Big)^2 \Big]
= \frac{n}{a^2} \int \beta^2 \, \mathrm{d}\rho + \frac{4}{a^2} \Delta^2 .
\]
\end{lemma}

\begin{proof}
\[
\E \sum_{l, l'} \beta(\bm{y}_l) \beta(\bm{y}_{l'})
= n \int \beta^2 \, \mathrm{d}\rho + n(n-1) \Big( \int \beta \, \mathrm{d}\rho \Big)^2
= n \int \beta^2 \, \mathrm{d}\rho,
\qquad
\E \sum_l \beta(\bm{y}_l) = 0;
\]
expand the square.
\end{proof}

\begin{lemma}[Two integrated moments]
\label{lem:app-moments}
For independent $X_0, X_1, Y \sim \rho$ and $\bm{u} = \Phi(X_0) - \Phi(X_1)$,
\[
\E \ip{\bar\Phi(Y)}{\bm{u}}^2 = 2\, T_\rho ,
\qquad
\E \ip{\bm{m}}{\bm{u}}^2 = 2\, V_\rho .
\]
\end{lemma}

\begin{proof}
For bounded measurable $\alpha$,
\[
\E \big( \alpha(X_0) - \alpha(X_1) \big)^2 = 2 \int \alpha^2 \, \mathrm{d}\rho - 2 \Big( \int \alpha \, \mathrm{d}\rho \Big)^2 .
\]
For the first claim take $\alpha = \bar k(\bm{y}, \cdot)$ at fixed $Y = \bm{y}$,
which integrates to $\ip{\bar\Phi(\bm{y})}{\int \bar\Phi \, \mathrm{d}\rho} = 0$,
then integrate $2 \int \bar k(\bm{y}, \bm{x})^2 \, \mathrm{d}\rho(\bm{x})$ over
$\bm{y}$ (Fubini for bounded integrands); for the second take
$\alpha = \ip{\bm{m}}{\bar\Phi(\cdot)}$, which integrates to $0$ and whose
square integrates to $V_\rho$.
\end{proof}

\begin{proposition}[Proposition~\ref{prop:gain}, full form]
\label{prop:app-gain}
Let $k$ be positive semidefinite with unit diagonal and $k \ge 0$, let
$\bm{z}_1, \dots, \bm{z}_M$ be independent with law $\rho$, $M \ge 2$, let
$\varepsilon \ge 0$, and let $\bm{w}_\varepsilon(\bm{z})$ be a ridged
constrained minimizer at the samples with
$\bm{z} \mapsto \MMD^2((\bm{z}, \bm{w}_\varepsilon(\bm{z})), \bm{m})$ measurable.
Then
\begin{equation}
\label{eq:app-sharper}
\begin{aligned}
&\E \big[ \MMD^2(Q_{\mathrm{eq}}, \bm{m}) - \MMD^2\big((\bm{z}, \bm{w}_\varepsilon(\bm{z})), \bm{m}\big) \big] \\
&\qquad\ge \frac{4 V_\rho + (M - 2)\, T_\rho}{(1 + \varepsilon)\, M^2}
\;\ge\; \frac{M-2}{(1+\varepsilon)M} \cdot \frac{T_\rho}{M}
= \frac{M-2}{(1+\varepsilon)M}\, r_\rho\, \frac{1 - c_\rho}{M},
\end{aligned}
\end{equation}
the second inequality for $M \ge 3$; the last member is \eqref{eq:gain}.
\end{proposition}

\begin{proof}
Write $M = n + 2$, order the samples as in Lemma~\ref{lem:app-decomp} with
the pair $(\bm{z}_1, \bm{z}_2)$, and integrate out the first two coordinates
of the product measure; the gain lies
in $[0, 4]$, so every integral below exists. By Lemma~\ref{lem:app-pair}
with $b = 2$, Lemmas~\ref{lem:app-decomp}--\ref{lem:app-condmoment} with
$a = M$, and Lemma~\ref{lem:app-moments},
\[
\text{gain} \ge \frac{\ip{D}{\bm{u}}^2}{2 + 2\varepsilon},
\qquad
\E\big[ \ip{D}{\bm{u}}^2 \,\big|\, \bm{z}_1, \bm{z}_2 \big]
= \frac{n}{M^2} \int \ip{\bar\Phi(\bm{y})}{\bm{u}}^2 \, \mathrm{d}\rho(\bm{y}) + \frac{4}{M^2} \ip{\bm{m}}{\bm{u}}^2,
\]
\[
\E \ip{D}{\bm{u}}^2 = \frac{2 (M-2)\, T_\rho + 8 V_\rho}{M^2};
\]
divide by $2 + 2\varepsilon$, and
drop $V_\rho \ge 0$ for the second inequality.
\end{proof}

For the deployed weights the measurability hypothesis holds:
$\bm{w}^\star$ of Proposition~\ref{prop:app-closedform} is a continuous
function of $(\bm{A}, \hat{\bm{\mu}})$ on the open set
$\{\bm{A} \text{ invertible},\ \bm{1}^{\!\top}\bm{A}^{-1}\bm{1} \ne 0\}$, which
contains every $\bm{K}(\bm{z}) + \varepsilon\bm{I}$;
$\bm{z} \mapsto (\bm{K}(\bm{z}), \bm{\mu}(\bm{z}))$ is measurable when $\Phi$ is
strongly measurable; and $\MMD^2((\bm{z}, \bm{w}^\star(\bm{z})), \bm{m})$ is a
finite sum of measurable functions, squared in norm.

\begin{lemma}[The range of $r_\rho$]
\label{lem:app-range}
$\int \|\bar\Phi\|^2 \, \mathrm{d}\rho = 1 - c_\rho$ and $0 \le T_\rho \le (1 - c_\rho)^2$;
if $c_\rho < 1$ then $T_\rho > 0$. Hence $0 < r_\rho \le 1 - c_\rho$ whenever
$c_\rho < 1$.
\end{lemma}

\begin{proof}
\[
\int \|\bar\Phi\|^2 \, \mathrm{d}\rho = \int \big( 1 - 2\ip{\bm{m}}{\Phi} + \|\bm{m}\|^2 \big) \, \mathrm{d}\rho = 1 - \|\bm{m}\|^2,
\qquad
T_\rho \le \iint \|\bar\Phi(\bm{x})\|^2 \|\bar\Phi(\bm{x}')\|^2 \, \mathrm{d}\rho \, \mathrm{d}\rho = (1 - c_\rho)^2 .
\]
For positivity let $C h = \int \ip{h}{\bar\Phi(\bm{y})}\, \bar\Phi(\bm{y}) \, \mathrm{d}\rho(\bm{y})$,
so $\ip{a}{Ch} = \int \ip{h}{\bar\Phi}\ip{a}{\bar\Phi} \, \mathrm{d}\rho = \ip{h}{Ca}$
and $\ip{h}{Ch} \ge 0$; if $T_\rho = 0$ then, for $\rho$-a.e.\ $\bm{x}$ and
every $a$, $t$,
\[
\ip{\bar\Phi(\bm{x})}{C\bar\Phi(\bm{x})} = \int \bar k(\bm{x}, \bm{y})^2 \, \mathrm{d}\rho(\bm{y}) = 0,
\quad
0 \le \ip{a + tv}{C(a + tv)} = \ip{a}{Ca} + 2t\ip{a}{Cv} \;\Rightarrow\; Cv = 0 \;\; (v = \bar\Phi(\bm{x}));
\]
hence for every $h$,
\[
\|Ch\|^2 = \int \ip{h}{\bar\Phi(\bm{x})}\ip{h}{C\bar\Phi(\bm{x})} \, \mathrm{d}\rho(\bm{x}) = 0,
\quad
\int \ip{h}{\bar\Phi}^2 \, \mathrm{d}\rho = \ip{h}{Ch} = 0,
\quad
\ip{h}{\bar\Phi(\bm{x})} = 0 \;\text{ for a.e.\ } \bm{x} .
\]
The range of the strongly measurable $\Phi$ is separable; $h$ over a
countable dense subset of it and continuity of the inner product give
$\|\bar\Phi(\bm{x})\|^2 = 0$ for a.e.\ $\bm{x}$, so
$1 - c_\rho = \int \|\bar\Phi\|^2 \, \mathrm{d}\rho = 0$.
\end{proof}

\begin{proposition}[No target-free constant]
\label{prop:app-notargetfree}
Let $\mathcal{X} = \mathbb{N}$, $k(i, j) = \mathbf{1}\{i = j\}$, and let
$\rho_n$ be uniform on $\{0, \dots, n-1\}$, so $c_{\rho_n} = 1/n$. For every
$M \ge 1$, $\varepsilon \ge 0$ and $n \ge 2$, every ridged constrained
minimizer with integrable gain satisfies
\[
\E\big[ \MMD^2(Q_{\mathrm{eq}}, \bm{m}) - \MMD^2\big((\bm{z}, \bm{w}_\varepsilon(\bm{z})), \bm{m}\big) \big]
\le \frac{4 M^2}{n},
\qquad
\frac{1 - c_{\rho_n}}{M} \ge \frac{1}{2M};
\]
hence for every $\beta > 0$ and $M$ there is $n$ with $c_{\rho_n} \le \tfrac12$
at which the expected gain is below $\beta\,(1 - c_{\rho_n})/M$.
\end{proposition}

\begin{proof}
$\Phi(i) = \bm{e}_i \in \ell^2(\mathbb{N})$ are orthonormal,
$\bm{m} = \tfrac1n \sum_{i < n} \bm{e}_i$, $\|\bm{m}\|^2 = 1/n$, and
$\ip{\Phi(i)}{\bm{m}} = \tfrac1n \mathbf{1}\{i < n\}$. On the event $E$ that
the samples are distinct atoms below $n$, for every unit-sum $\bm{w}$,
\[
\MMD^2\big((\bm{z}, \bm{w}), \bm{m}\big) = \|\bm{w}\|^2 - \frac1n \ge \frac1M - \frac1n = \MMD^2(Q_{\mathrm{eq}}, \bm{m}),
\]
so the gain is at most $0$ on $E$, and at most $4$ off $E$. With $N$ the
number of ordered pairs $i \ne j$ with $\bm{z}_i = \bm{z}_j$ plus the number of
samples $\ge n$,
\[
\text{gain} \le 4 N, \qquad
\E N = \frac{M(M-1)}{n} + 0 \le \frac{M^2}{n}, \qquad
\frac{1 - c_{\rho_n}}{M} = \frac{1 - 1/n}{M} \ge \frac{1}{2M},
\qquad
n > \frac{8M^3}{\beta} \;\Rightarrow\; \frac{4M^2}{n} < \frac{\beta}{2M} .
\qedhere
\]
\end{proof}

\subsection{Theorem~\ref{thm:neverworse}(ii): the estimated kernel mean}
\label{app:estimated}

Let $X_1, \dots, X_L$, $L \ge 2$, be the certification split, independent
with law $\rho$ and independent of everything used to build the candidates,
which are fixed below; $\hat{\bm{\mu}}(\bm{z})_j = \tfrac1L \sum_i k(\bm{z}_j, X_i)$
and $\widehat J$ is computed with it. The witness of $Q$ is
$h_Q(\bm{x}) = \sum_j w_j k(\bm{z}_j, \bm{x}) = \ip{\Phi_Q}{\Phi(\bm{x})}$,
and $\kappa$ bounds $\|\Phi\|$; the body
has $\kappa = 1$. For candidates $a, b$: $g_{ab} = h_{Q_a} - h_{Q_b}$,
$D_{ab} = \|\Phi_{Q_a} - \Phi_{Q_b}\| = \MMD(Q_a, Q_b)$,
$\widehat\Delta_{ab} = \widehat J(Q_a) - \widehat J(Q_b)$,
$\Delta_{ab} = \MMD^2(Q_a, \bm{m}) - \MMD^2(Q_b, \bm{m})$.

\begin{lemma}[The engine identity]
\label{lem:app-engine}
At every realization of the split,
\[
\widehat J(Q) - \big(\MMD^2(Q, \bm{m}) - c_\rho\big) = -\frac{2}{L} \sum_{i=1}^L \big( h_Q(X_i) - \ip{\Phi_Q}{\bm{m}} \big),
\qquad
\widehat\Delta_{ab} - \Delta_{ab} = -\frac{2}{L} \sum_{i=1}^L \big( g_{ab}(X_i) - \E\, g_{ab}(X) \big).
\]
\end{lemma}

\begin{proof}
The Gram term is common to $\widehat J$ and to $\MMD^2 - c_\rho$
(Lemma~\ref{lem:app-crit}), and
\[
\bm{w}^{\!\top} \hat{\bm{\mu}}(\bm{z}) = \frac1L \sum_i h_Q(X_i),
\qquad
\bm{w}^{\!\top} \bm{\mu}(\bm{z}) = \ip{\Phi_Q}{\bm{m}} = \E\, h_Q(X);
\]
subtract the two per-candidate identities.
\end{proof}

\begin{lemma}[The witness bounds]
\label{lem:app-witness}
For every $\bm{x}$, $|g_{ab}(\bm{x})| \le D_{ab} \|\Phi(\bm{x})\| \le D_{ab}\, \kappa$, and
$D_{ab}^2 = \bm{w}_a^{\!\top} \bm{K}(\bm{z}_a) \bm{w}_a - 2 \bm{w}_a^{\!\top} \bm{K}(\bm{z}_a, \bm{z}_b) \bm{w}_b + \bm{w}_b^{\!\top} \bm{K}(\bm{z}_b) \bm{w}_b$.
\end{lemma}

\begin{proof}
$g_{ab}(\bm{x}) = \ip{\Phi_{Q_a} - \Phi_{Q_b}}{\Phi(\bm{x})}$, Cauchy--Schwarz,
and the expansion of $\|\Phi_{Q_a} - \Phi_{Q_b}\|^2$.
\end{proof}

\begin{lemma}[Selection on the concentration event]
\label{lem:app-selection}
Let $\gamma \ge 0$ with $|\widehat\Delta_{ab} - \Delta_{ab}| \le \gamma$. If
$\widehat J(Q_a) \le \widehat J(Q_b)$ then
$\MMD^2(Q_a, \bm{m}) \le \MMD^2(Q_b, \bm{m}) + \gamma$, and the excess of $Q_a$
over the better of the two is at most $\gamma$ and is zero when the choice is
correct. If moreover $|\widehat\Delta_{ab}| > \gamma$, the candidate with the
smaller $\widehat J$ has the smaller discrepancy.
\end{lemma}

\begin{proof}
\[
\Delta_{ab} \le \widehat\Delta_{ab} + \gamma \le \gamma;
\qquad
|\widehat\Delta_{ab}| > \gamma \;\Rightarrow\; 0 \notin [\widehat\Delta_{ab} - \gamma, \widehat\Delta_{ab} + \gamma] \ni \Delta_{ab} .
\qedhere
\]
\end{proof}

\begin{lemma}[Coincident candidates]
\label{lem:app-degenerate}
If $\Phi_{Q_a} = \Phi_{Q_b}$, as for an untrained network, then $g_{ab} \equiv 0$,
$D_{ab} = 0$, the unbiased sample variance $\widehat V_{ab}$ of $(g_{ab}(X_i))_i$
is $0$,
and $\hat\gamma_{ab}(\delta) = 0$ for every $\delta$.
\end{lemma}

\begin{proof}
$g_{ab}(\bm{x}) = \ip{\Phi_{Q_a} - \Phi_{Q_b}}{\Phi(\bm{x})}$ and \eqref{eq:slack}.
\end{proof}

\begin{lemma}[Empirical Bernstein bound]
\label{lem:app-mp}
Let $Z_1, \dots, Z_n$, $n \ge 2$, be independent and identically distributed
with values in $[0,1]$, mean $\E Z$, empirical mean $\bar Z$ and unbiased
sample variance $V_n$. For $\delta' \in (0,1)$, with probability at least
$1 - \delta'$,
$\E Z - \bar Z \le \sqrt{2 V_n \ln(2/\delta')/n} + 7\ln(2/\delta')/(3(n-1))$.
\end{lemma}

\noindent Theorem~4 of \citet{maurer2009empirical}.

\begin{theorem}[Theorem~\ref{thm:neverworse}(ii), full form]
\label{thm:app-certified}
Let $\mathcal{A} = \{Q_0, Q_{\mathrm{rw}}, Q_{\mathrm{mv}}\}$ be scored on
the split, let $\hat a$ minimize $\widehat J$, let $\delta \in (0,1)$, and let
$\hat\gamma_{ab}(\delta)$ be the slack \eqref{eq:slack} with $D_{ab}\kappa$
in place of $D_{ab}$, $\widehat V_{ab}$ the unbiased sample variance of
$(g_{ab}(X_i))_i$. On one event of probability at least $1 - \delta$
over the split,
\[
\begin{aligned}
\MMD^2(Q_{\hat a}, \bm{m}) &\le \MMD^2(Q_0, \bm{m}) + \hat\gamma_{\hat a, 0}(\delta/3), \\
\MMD^2(Q_{\hat a}, \bm{m}) &\le \min_{a} \MMD^2(Q_a, \bm{m}) + \max_{a \ne b} \hat\gamma_{ab}(\delta/3),
\end{aligned}
\]
and between $Q_{\mathrm{rw}}$ and $Q_{\mathrm{mv}}$ the candidate with the
smaller $\widehat J$ has the smaller discrepancy whenever
$|\widehat\Delta_{\mathrm{rw},\mathrm{mv}}| > \hat\gamma_{\mathrm{rw},\mathrm{mv}}(\delta/3)$,
and otherwise the excess over the better of them is at most that slack.
\end{theorem}

\begin{proof}
For each pair $a \ne b$ with $D_{ab} > 0$ (Lemma~\ref{lem:app-degenerate}
covers $D_{ab} = 0$), Lemma~\ref{lem:app-mp} at level $\delta/6$, applied to
$Z_i = (g_{ab}(X_i) + D_{ab}\kappa)/(2 D_{ab}\kappa) \in [0,1]$
(Lemma~\ref{lem:app-witness}) and to $1 - Z_i$, with
$V_n(Z) = \widehat V_{ab}/(4 D_{ab}^2 \kappa^2)$ and Lemma~\ref{lem:app-engine},
gives an event $E_{ab}$ of probability at least $1 - \delta/3$ on which
\[
|\widehat\Delta_{ab} - \Delta_{ab}| = 2 |\bar g_{ab} - \E g_{ab}|
\le 2\sqrt{\frac{2 \widehat V_{ab} \ln(12/\delta)}{L}} + \frac{28}{3}\, \frac{D_{ab}\kappa \ln(12/\delta)}{L - 1}
= \hat\gamma_{ab}(\delta/3),
\]
and $\Pr( E_{\mathrm{rw},0} \cap E_{\mathrm{mv},0} \cap E_{\mathrm{rw},\mathrm{mv}} ) \ge 1 - \delta$.
On the intersection, Lemma~\ref{lem:app-selection} applied to the pairs
$(\hat a, 0)$, $(\hat a, a^\star)$ with $a^\star$ a true minimizer, and
$(\mathrm{rw}, \mathrm{mv})$ gives the three claims.
\end{proof}

\begin{corollary}[The Hoeffding form]
\label{cor:hoeffding}
For fixed candidates $a, b$ and $\delta \in (0,1]$,
\[
\Pr\Big( |\widehat\Delta_{ab} - \Delta_{ab}| > 2 D_{ab}\,\kappa \sqrt{2 \ln(2/\delta)/L} \Big) \le \delta,
\]
so Theorem~\ref{thm:app-certified} holds with
$2 D_{ab}\kappa\sqrt{2\ln(6/\delta)/L}$ in place of $\hat\gamma_{ab}(\delta/3)$.
\end{corollary}

\begin{proof}
With $Y_i = g_{ab}(X_i) - \E g_{ab}(X)$, independent and centered, each a
bounded variable of range $[-D_{ab}\kappa, D_{ab}\kappa]$ shifted by its mean,
Hoeffding's lemma and inequality give, at $t = D_{ab}\kappa\sqrt{2\ln(2/\delta)/L}$,
\[
|\widehat\Delta_{ab} - \Delta_{ab}| = \frac2L \Big| \sum_i Y_i \Big|,
\qquad
\Pr\Big( \pm \sum_i Y_i \ge L t \Big) \le \exp\Big( -\frac{(Lt)^2}{2 L D_{ab}^2 \kappa^2} \Big) = \frac{\delta}{2};
\]
the event $|\widehat\Delta_{ab} - \Delta_{ab}| \le 2t$ replaces $E_{ab}$ in
the proof of Theorem~\ref{thm:app-certified}.
\end{proof}

\section{Architecture, training and the problems}
\label{app:setup}

\subsection{Architecture and training}
\label{app:arch}

The quadrature field is a stack of permutation-equivariant blocks, each a
multi-head self-attention layer over the $M$ samples and a two-layer
perceptron per sample, both residual with a layer normalization over
channels. A feature-wise affine modulation~\citep{perez2018film} by an
embedding of $\log M$ and of the log spread of the set enters every block. The
output layer is zero-initialized and scaled by one half. On the conditional
problems each block adds a residual cross-attention from the samples to tokens
from an encoder of the observation, with a zero-initialized output projection.
The encoder is a strided convolutional stack on the filtered back-projection
image for the tomography flow, and a two-layer perceptron to eight tokens for
the groundwater sensors and the Lotka--Volterra summary. The closed forms and
the exact tomography posterior carry no conditioner.

Training minimizes \eqref{eq:loss} in double precision, with the weights of
\eqref{eq:weights} solved at ridge $10^{-10}$ and differentiated through, by Adam~\citep{kingma2015adam} at a node count drawn uniformly from $\{32, \dots, 512\}$ at each step.
Table~\ref{tab:arch} gives the rest.

\begin{table}[h]
\centering
\caption{Architecture and training, by problem.}
\label{tab:arch}
\small
\begin{tabular}{lccccrr}
\toprule
& width, depth, heads & conditioner & batch & kernel mean & rate & steps \\
\midrule
closed forms      & 64, 3, 4  & none                  & 64 sets         & exact        & $3 \cdot 10^{-4}$ & 4000 \\
exact posterior   & 64, 3, 4  & none                  & 64 sets         & exact        & $3 \cdot 10^{-4}$ & 4000 \\
tomography flow   & 128, 4, 8 & back-projection image & 32 observations & stored dataset  & $3 \cdot 10^{-4}$ & 4000 \\
groundwater       & 96, 3, 4  & 33 sensors            & 16 observations & fresh dataset   & $3 \cdot 10^{-4}$ & 4000 \\
Lotka--Volterra   & 64, 8, 4  & 43 summaries          & 8 observations  & fresh dataset   & $10^{-4}$         & 87400 \\
\bottomrule
\end{tabular}
\end{table}

The two-component mixtures and the banana run 3500 steps. A fresh dataset is
2048 samples redrawn at every step, alongside $M$ fresh samples per
observation; the stored dataset is 512 samples per observation, the seeds taken
from it without replacement and the kernel mean estimated on a fixed
256-sample subset. Groundwater stops early on 64 held-out observations scored
at $M = 128$ every 200 steps with patience eight. Lotka--Volterra clips
gradients at five, validates on 64 held-out observations at the three
smallest node counts every 200 steps with patience thirty, and stopped at
87400 of 200000 steps.

\subsection{The closed-form references}
\label{app:closed}

The two-dimensional Gaussian has prior covariance
$\big[\begin{smallmatrix} 1.5 & 0.2 \\ 0.2 & 1.0 \end{smallmatrix}\big]$,
forward map $\big[\begin{smallmatrix} 1.0 & 0.3 \\ -0.1 & 0.9 \end{smallmatrix}\big]$
and noise covariance $\operatorname{diag}(0.15, 0.2)$. The ten-dimensional
mixture has two components with covariances $0.5\bm{I}$ and $2\bm{I}$ and
centers $8\sqrt{2}$ apart. The banana is $x_1 \sim \mathcal{N}(0, 3^2)$,
$x_2 = 0.1(x_1^2 - 9) + t$, $t \sim \mathcal{N}(0, 1)$. Its reference is a dataset
of 160000 samples, and the kernel mean is summed over a 32000-sample
subset in training and over the full dataset in evaluation. Nine further
closed forms enter the split between displacement and reweighting: Gaussians
in dimensions 4 to 32 with prior
$\bm{Q}\operatorname{diag}(1, \dots, 2)\bm{Q}^{\!\top}$, identity forward map
and noise 0.3, and five-mode and two-component mixtures in dimensions 10, 32
and 64. Their bandwidth is the median heuristic on 4000 samples of each
reference.

\subsection{Limited-angle tomography}
\label{app:tomo}

The exact posterior is on a $16 \times 16$ image under the Whittle--Mat\'ern
prior $\tau^2(\delta \bm{I} - \Delta)^{-\alpha}$ with $\delta = 0.5$,
$\alpha = 2$, $\tau = 1$ and reflecting boundary, observed at 16 angles in
$[-75^\circ, 75^\circ]$ and 16 offsets with noise $\mathcal{N}(0, 0.8^2 \bm{I})$.
Its Gaussian posterior is projected onto the likelihood-informed
subspace~\citep{cui2014lis}, the leading eigenvectors of the prior-whitened
Fisher information, at the rank set by the participation ratio of that
spectrum under a cap of 32, which is seven.

The trained flow is on a $64 \times 64$ image at 60 angles in
$[-60^\circ, 60^\circ]$ and 90 detector bins, with Gaussian sinogram noise
fixed per sample at 40 dB, conditioned on the Ram--Lak filtered
back-projection. It has eight coupling trees at 32 hidden channels and
dequantization 0.02, trained by Adam at $1.5 \cdot 10^{-4}$ and batch 16 to
the best-validation checkpoint, on 256 observations with 512 samples each and
64 held out. The field reads it in the principal subspace of the pooled
training samples, whitened per observation.

\subsection{Groundwater flow}
\label{app:gw}

The pressure equation is solved on a $40 \times 40$ vertex grid by a sparse
direct solve, pressure fixed on the top and bottom edges and zero flux on the
sides. The log-permeability is a Karhunen--Lo\`eve expansion over $10 \times 10$
modes with amplitudes $(\pi^2 |\bm{k} + \tfrac12|^2)^{-0.55}$ and
$\bm{\xi} \sim \mathcal{N}(0, \bm{I}_{100})$, observed at 33 sensors, one
central and 32 on a circle of radius $39/80$, with noise 0.01; the linearized
operator resolves ten of the hundred prior directions. The reference is a
HINT flow~\citep{kruse2021hint} fitted by conditional maximum likelihood to
280000 of 300000 simulated pairs, the last 20000 held out. The field draws
512 observations from the flow's training rows, 64 of them held out for
validation. Each evaluation cell uses an 8192-sample dataset per observation,
doubled at the two largest node counts.

The verification of Section~\ref{sec:depth} runs four preconditioned
Crank--Nicolson chains per observation, from overdispersed starts, on
eight held-out observations at $M = 64$, long past burn-in, and the kept
samples agree across chains and halves. The flow's distance to the
chains is at the chains' agreement across halves, so the flow
approximates the posterior as accurately as the chains resolve. The
returned quadrature is below the floor against the chains on every
observation.

\subsection{The Lotka--Volterra family}
\label{app:lv-family}

The forward model is that of \citet{riabiz2022steinthinning}:
$\dot u_1 = \theta_1 u_1 - \theta_2 u_1 u_2$, $\dot u_2 = \theta_4 u_1 u_2 - \theta_3 u_2$
from $u(0) = (1, 1)$ on $[0, 25]$, both species observed at 2400 equally
spaced times under Gaussian noise of standard deviation 0.2, with unknown
$\bm{x} = \log\bm{\theta} \in \mathbb{R}^4$ under a standard Gaussian prior.
The state is integrated in logarithms, and a trajectory is accepted by an
explicit predicate applied identically in the simulator and in every
likelihood scored against it; of $10^5$ prior samples all but three are
accepted. The test family is stratified: 400 candidates are ordered by the
geometric mean of their Laplace posterior standard deviations, eight are taken
at evenly spaced quantiles from 0.03 to 0.97, and the generating parameter of
\citet{riabiz2022steinthinning}, $\bm{\theta}^\star = (0.67, 1.33, 1, 1)$, is
appended as a ninth. Their widths span a factor of 347; the sharpest member
completes twelve prey cycles in the window and the broadest none.

The conditioner is 43 numbers computable from the observation. Fourteen are
summary statistics of the trajectory and fifteen a Gauss--Newton refinement
started from a ridge prediction on them: its iterate, the log-Cholesky factor
of its precision, and the logarithm of its chi-square per degree of freedom.
Fourteen more carry the whitening frame taken from each member's conditional
moments. A chi-square per degree of freedom above two, a test on the
observation alone fixed before any reference was built, excludes one of the
nine members. The reference is a HINT flow of six coupling layers at width 64
on this conditioner, fitted to 41952 members.

The flow is verified per member against four converged chains,
preconditioned by the linearized covariance. On the sharpest and the
broadest members the flow's distance to the chains exceeds the floor, so
on this family the reference is not the exact posterior at the accuracy
the quadrature reaches, and every claim is against the reference
(Remark~\ref{rem:reference}).

The downstream functionals of Section~\ref{sec:peers} push every node through one
solve of the system, extended to $t = 35$ at the grid spacing of the
observations. The exceedance is a logistic in the windowed prey peak, with
the threshold at the 0.9 quantile and the scale at half the standard
deviation of a 4096-solve calibration block drawn before the gold
reference. The gold holds 131072 solves per member in two agreeing halves,
and the exceedance is registered in advance as the primary functional.
The peak time is refined below the grid by the parabola through the three
surrounding points, and the broadest member, whose typical trajectory has no
interior peak, is censored from it. Errors are normalized per member by the
gold standard deviation, cells the gold cannot separate are reported as
unresolved under one Holm correction, and the peak-time multipliers are
floors, limited by the gold's own resolution. Every comparison extends
the seed rule of the shipped evaluation, whose selections its first twelve
repeats reproduce in all forty member--node cells, and the panel costs 1.7
million solves, two hours on one CPU.

\subsection{Protocol, seeds and compute}
\label{app:seeds}

Table~\ref{tab:protocol} gives where each quadrature runs and what a reported
cell averages; the bandwidth on the closed forms is the median heuristic.

\begin{table}[h]
\centering
\caption{Coordinates, protocol and size, by problem.}
\label{tab:protocol}
\footnotesize
\setlength{\tabcolsep}{4.5pt}
\begin{tabular}{lcccccr}
\toprule
& coordinates & bandwidth & observations & a cell averages & parameters \\
\midrule
closed forms    & ambient            & median  & one            & 256 sample sets        & 0.11--0.12\,M \\
exact posterior & rank-7 informed    & 1.4093  & one held out   & 256 sample sets        & 0.11--0.12\,M \\
tomography flow & rank-32 principal  & $\sqrt{32}$ & 64 held out & a 1024-sample dataset & 1.04\,M \\
groundwater     & 100 KL modes       & 14.1489 & 64 afresh      & 32 sample sets         & 0.51\,M \\
Lotka--Volterra & ambient            & 2.5544  & eight members  & twelve sample sets     & 0.63\,M \\
\bottomrule
\end{tabular}
\end{table}

The fit's break-even is measured on the closed-form
references: its own cost is repaid against a per-observation descent
after a few hundred to a few thousand queries across a range of exchange
rates between samples and kernel-mean evaluations. The intervals span a
decade, and each count is a lower bound since the reference's own training
is not charged.

The base seed 20260612 fixes the data and the held-out samples of every run,
and the training seed is 0 wherever no band over training seeds is shown.
Bands appear on the closed forms, three retrains with the interquartile
range. The tomography flow is retrained in three full pipelines, whose
spread the record keeps for the quadrature and the floor separately, so its
curve is drawn without a band. The banana's
retrains use a 20000-sample dataset with a 4000-sample subset. Every other
run is trained once, and its median is over evaluation randomness. The
Lotka--Volterra peers run four replicates. All
quadrature computations run in double precision on the CPU, an Intel Xeon
w3-2525 at four threads per run; the tomography and groundwater flows are
trained and sampled on a 16\,GB RTX 2000 Ada.
\end{document}